\documentclass[twoside,11pt]{article}
\usepackage{amsmath,amssymb,amsthm,amsfonts,mathrsfs,wasysym,latexsym,times,lineno, subfigure,color,booktabs}
\usepackage{amsmath,amsfonts}
\usepackage{multirow}
\usepackage{array}
\usepackage{amsthm}
\usepackage{graphicx}
\usepackage{color}
\usepackage{multicol}
\usepackage{float}

\newtheorem{thm}{Theorem}

\newtheorem{definition}{Definition}[section]

\newtheorem{lemma}[thm]{Lemma}

\newtheorem{theorem}[thm]{Theorem}

  \date{}

\title{Super Solutions of the Model RB}

\author{Guangyan Zhou$^{1}$,\quad Wei Xu $^{2,*}$\\
\footnotesize $^{1}$Department of Mathematics,
Beijing Technology and Business University, Beijing, 100048, China
\\
\footnotesize $^2$ School of Mathematics and Physics, University of Science and Technology Beijing, Beijing 100083, China\\}

\begin{document}
\maketitle

\begin{abstract}
The concept of super solution is a special type of generalized solutions with certain degree of robustness and stability. In this paper we consider the $(1,1)$-super solutions of the model RB. Using the first moment method, we establish a ``threshold'' such that as the constraint density crosses this value, the expected number of $(1,1)$-super solutions goes from $0$ to infinity.

\end{abstract}
\footnotetext[1]{ Corresponding  author. xuwei@ustb.edu.cn.
This work was supported by National Natural Science Foundation of China (61702019, 11801028). }
\section{Introduction}

In many combinatorial optimization and decision problems, what people concern is to find solutions of minimal costs. In practice, however, such optimal solutions can be very brittle in that if the value of one variable becomes unavailable, repairing this solution leads to a great increase in its final cost. Therefore, the concept of super solution is introduced to formalize a solution with a certain degree of robustness or stability. To quantify the robustness, $(a,b)$-super solution was introduced to constraint programming in \cite{H}. Specifically, an $(a,b)$-super solution is one in which if the values assigned to $a$ variables are no longer available, the solution can be repaired by assigning these variables with $a$ new values and at most $b$ other variables.

Over the past years, random models of constraint satisfaction problems (CSPs) have been intensively studied. Initially, four ``standard" models known as models A, B, C and D \cite{SD,GMPS} have been introduced to generate random binary CSP instances. However, Achlioptas et. al. \cite{AKK} have proved that random instances of these models suffer from unsatisfiability as the number of variables increases. Overcoming the trivial unsatisfiability, Xu and Li proposed Model RB \cite{xu}, which is a prototype CSP model with growing domains revised from the standard CSP model B \cite{SD}.
It has been shown that random instances of model RB are hard to solve at the threshold both theoretically \cite{xu} and experimentally \cite{xu06}, thus benchmarks based on model RB have been widely used in algorithm competitions.  In addition, studies on the statistical mechanics of model RB  show that the replica symmetry solution is always stable in the satisfiable phase \cite{zhao12,zhao11}, and there is no condensation transition in model RB \cite{xuwei}.

An instance  $\mathbf{I}$ of Model RB contains a finite set $V=\{x_1,...,x_n\}$ of $n$ variables and a set of constraints $\mathcal{C}=\{C_1,...,C_m\}$ ($m=rn\ln n$ and $r>0$ is a constant). Each $x_i$ takes values from domain $D=\{1,2,...,d\}$, where $d=n^\alpha$ ($\alpha>0$ is a constant). Each constraint $C_i$ is consisted of a set of $k$ variables $X_i=\{x_{i_1},x_{i_2},...,x_{i_k}\}$, where the $k$ ($k\ge2$) different variables are chosen uniformly at random, and a set of satisfying tuples of values $R_i\subset D^k$ with $|R_i|=pd^k$ ($0<p<1$ measures the tightness of a constraint). Constraint $C_i$ is satisfied if the tuple of values assigned to $X_i$ is in the corresponding relation $R_i$. Solving the instance is to find an assignment of $V$ that satisfies all the constraints or to prove that no such assignment exists.

In \cite{ZG}, Zhang and Gao studied the $(1,0)$-super solutions of model RB where they consider the special case $k=2$ and $d=\sqrt n$. By using the first moment method, they derived an upper bound of the threshold of having a $(1,0)$-super solution asymptotically with probability 1, and established a condition for the expected number of super solutions to grow exponentially. $(1,0)$-super solutions have also been studied in random (3+p)-SAT \cite{wangzhou19}.

 Letting $\triangle(\sigma,\tau)$ denote the set of variables being assigned different values by $\sigma$ and $\tau$, i.e., $\triangle(\sigma,\tau)=\{x_i|\sigma(x_i)\ne\tau(x_i)\}$, we can give the following definitions of $(1,0)$ and $(1,1)$-super solutions.

\begin{definition}\label{def1}(\textbf{$(1,0)$-super solution}).
 An assignment $\sigma\in D^n$ is called a $(1,0)$-super solution if $\sigma$ is a satisfying assignment, and for every variable $x_i$ ($1\le i\le n$), there exists another satisfying assignment $\tau$ such that $\triangle(\sigma,\tau)=\{x_i\}$.
\begin{definition}\label{def2}(\textbf{$(1,1)$-super solution}).
An assignment $\sigma\in D^n$ is called a $(1,1)$-super solution if $\sigma$ is a satisfying assignment, and for every variable $x_i$ ($1\le i\le n$), there exists another satisfying assignment $\tau$ such that either $\triangle(\sigma,\tau)=\{x_i\}$, or there exists $j\ne i$ such that $\triangle(\sigma,\tau)=\{x_i,x_j\}$.
\end{definition}
\end{definition}

In this paper, we consider the $(1,1)$-super solution of model RB, and by using the first moment method, we prove rigorously the following results.

\begin{theorem}\label{11}
Let $Y$ be the number of $(1,1)$-super solutions of the model RB. Then
  \begin{equation*}\label{eq:g,e}
  \mathbf{E}[Y]=\left\{
 \begin{aligned}
 &0, \text{\quad\quad if }r>-\frac{\alpha}{\ln p},\\
&+\infty, \text{ if }r<-\frac{\alpha}{\ln p}.
 \end{aligned}
 \right.
 \end{equation*}
\end{theorem}
From the definitions, we can see that $(1,0)$ and $(1,1)$-super solutions reflect certain characteristics of structures of solution space, since they can be viewed as a particular type of ``locally maximum'' subset of satisfying solutions. In addition, the set of $(1,0)$-super solutions is a subset of $(1,1)$-super solutions, and the set of $(1,1)$-super solutions is a subset of standard solutions. This explains why the upper bound on $r$ of having $(1,0)$-super solutions in Theorem 5 in \cite{ZG} is smaller than that of having $(1,1)$-super solutions in Theorem \ref{11} and standard solutions obtained in \cite{xu}.

Our results are interesting in several aspects. First, there seems to be a sharp phase transition for the $(1,1)$-super solution where the expected number of $(1,1)$-super solutions goes from $0$ to infinity, and this threshold is exactly the same with the standard satisfiability threshold of model RB in \cite{xu}. Second, the restrictions of $r$ in Theorem \ref{11} is independent of the constraint size $k$. In addition, from the Definitions \ref{def1} and \ref{def2}, super solutions can be viewed as a certain type of ``locally maximum" subset of the standard solutions, thus we believe our results reflect from another aspect the structure of solution space and stability of solutions. It is worth mentioning that whether $-\frac{\alpha}{\ln p}$ in Theorem \ref{11} is a sharp threshold around which the model RB goes from being a.s. (almost surely) $(1,1)$-satisfiable to a.s. $(1,1)$-unsatisfiable is still open, and the second moment method will be needed to solve this problem.

\section{Proof of Theorem \ref{11}}
Let $\sigma$ be  a fixed assignment, $\mathbf{I}$ be a random instance of model RB, and $\sigma\models \mathbf{I}$ be that $\sigma$ satisfies $\mathbf{I}$. We consider the following events which was similarly used in \cite{ZG}:
\begin{itemize}
  \item $S(\sigma)$: $\sigma$ is a solution for $\mathbf{I}$.
  \vspace{-5pt}
  \item $R_i(\sigma)$: there exists another solution $\tau$ for $\mathbf{I}$ such that $\triangle(\sigma,\tau)=\{x_i\}$, or there exists a $j\ne i$ such that $\triangle(\sigma,\tau)=\{x_i,x_j\}$.
       \vspace{-5pt}
  \item $W(\sigma)$: $\sigma$ is a $(1,1)$-super solution for $\mathbf{I}$.
\end{itemize}

It is straightforward to see that
\begin{eqnarray*}
\mathbf{P}(W(\sigma))=\mathbf{P}\left(S(\sigma)\bigcap_{1\leq i\leq n}R_i(\sigma)\right).
\end{eqnarray*}

Following the notations  in \cite{ZG}, let $\mathcal{M}\subset (\underbrace{X\times ...\times X}_k)^m$ be the collection of all possible multi-sets of $m$ unordered constraints ($|\mathcal{M}|=\binom {n} {k}^{m}$). For a given set $e\in \mathcal{M}$ of $m$ unordered tuples, denote by $E(e)$ the event that $e$ is selected as the set of constraints of the random instance. By the definition of the random model, the set of events $\{E(e),e\in \mathcal{M}\}$ are mutually disjoint. Also, for any $e\in \mathcal{M}$, $\mathbf{P}(E(e))=\binom {n} {k}^{-m}$ and $\mathbf{P}(S(\sigma)|E(e))=p^m$. The probability of $\sigma$ being a $(1,1)$-super solution works out to be

\begin{align}\label{Tsigma}
\nonumber\mathbf{P}(W(\sigma))&=\mathbf{P}\left(S(\sigma)\cap_{1\leq i\leq n}R_i(\sigma)\right)\\
\nonumber&=\mathbf{P}\left(\cup_{e\in \mathcal{M}}\left(E(e)\cap S(\sigma)\cap(\cap_{1\le i\le n}R_i(\sigma))\right)\right)\\
\nonumber&=\sum_{e\in \mathcal{M}}\mathbf{P}\left(E(e)\cap S(\sigma)\cap(\cap_{1\le i\le n}R_i(\sigma))\right)\\
\nonumber&=\sum_{e\in \mathcal{M}}\mathbf{P}(E(e))\mathbf{P}(S(\sigma)|E(e)\mathbf{P}\left(\cap_{1\le i\le n}R_i(\sigma)|S(\sigma)\cap E(e)\right)\\
&=\binom {n} {k}^{-m}p^m\sum_{e\in \mathcal{M}}\mathbf{P}\left(\cap_{1\le i\le n}R_i(\sigma)|S(\sigma)\cap E(e)\right).
\end{align}

By inclusion$/$exclusion, we have
\begin{eqnarray}\label{exclusion}
\nonumber&&1-\sum_{i=1}^n\mathbf{P}(\overline{R_i(\sigma)}|S(\sigma)\cap E(e))\\
&\le&\mathbf{P}\left(\cap_{1\le i\le n}R_i(\sigma)|S(\sigma)\cap E(e)\right)
\\ \nonumber&\le&1-\sum_{i=1}^n\mathbf{P}(\overline{R_i(\sigma)}|S(\sigma)\cap E(e))+\sum_{1\le i<j\le n}\mathbf{P}(\overline{R_i(\sigma)}\cap\overline{R_j(\sigma)}|S(\sigma)\cap E(e)).
\end{eqnarray}

In the following, assume that $\sigma$ is a solution of $\mathbf{I}$, and $\mathcal{C}(x_i)$ is the set of constraints containing the variable $x_i$, $1\le i\le n$. Let  $D_i=D\backslash \{\sigma(x_i)\}$, and $q=1-p$ be the probability that a random constraint can not be satisfied by an assignment. 

\begin{lemma}\label{sigma1}
Let $\sigma$ be a given solution of model RB. For any fixed $y\in D_i$, suppose an assignment $\tau$ assigns $x_i$ the value $y$ and  $|\triangle(\sigma,\tau)|=1$ or $2$. For any constraint $C\in\mathcal{C}(x_i)$,
\begin{align*}
\mathbf{P}[C\text{ is satisfied by }\tau]=\rho\equiv1-q^{1+(d-1)(k-1)}.
\end{align*}
\end{lemma}

\begin{proof}

Note that $\sigma$ is a satisfying solution of $\mathbf{I}$, $|\triangle(\sigma,\tau)|=1$ or $2$, $\tau(x_i)=y\ne\sigma(x_i)$, thus there may exists another variable being assigned different values by $\sigma$ and $\tau$.

If $|\{x\in C|\sigma(x)\ne\tau(x)\}|=1$ (indeed, it is $x_i$), there is exactly $1$ way for $\tau$ to assign variables in $C$; if $|\{x\in C|\sigma(x)\ne\tau(x)\}|=2$, there are $(d-1)(k-1)$ ways for $\tau$ to assign variables in $C$. Since each way falsifies $C$ with probability $q$, thus $\mathbf{P}[\tau\nvDash C]=q^{1+(d-1)(k-1)}$, and
$\mathbf{P}[\tau\models C]=1-q^{1+(d-1)(k-1)}=\rho.$
\end{proof}

From Lemma \ref{sigma1}, let $\tau$ be an assignment such that $\triangle(\sigma,\tau)=\{x_i\}$ or $\{x_i,x_j\}$ for some $j\ne i$, and $\tau(x_i)=y$ ($y\in D_i$), we see that

\begin{align}\label{exleft}
\nonumber\mathbf{P}(\overline{R_i(\sigma)}|S(\sigma)\cap E(e))
&=\mathbf{P}(\cap_{y\in D_i}\overline{S(\tau)}|S(\sigma)\cap E(e))\\
\nonumber&=\prod_{y\in D_i}\mathbf{P}(\overline{S(\tau)}|S(\sigma)\cap E(e))\\
\nonumber&=\left[1-\mathbf{P}(\text{Constraints in $\mathcal{C}(x_i)$ are all satisfied by }\tau)\right]^{d-1}\\
&=\left(1-\rho^{m_i}\right)^{d-1}.
\end{align}

\begin{lemma}\label{intersection}
Suppose $y\in D\backslash \{\sigma(x_i)\}$, $z\in D\backslash \{\sigma(x_j)\}$. If $\tau, \omega$ be assignments that $\tau(x_i)=y$ and  $|\triangle(\sigma,\tau)|=1$ or $2$; $\omega(x_j)=z$ and  $|\triangle(\sigma,\omega)|=1$ or $2$. Then for any  $C\in\mathcal{C}(x_i)\cap\mathcal{C}(x_j)$,
\begin{align*}
\mathbf{P}[\tau,\omega\models C]
=1-2q^{1+(d-1)(k-1)}+q^{1+2(d-1)(k-1)}.
\end{align*}
\end{lemma}

\begin{proof}
From the proof of Lemma \ref{sigma1}, there are $1+(d-1)(k-1)$ possible ways for $\tau$ to assign variables in $C$, and  let $A(\tau)$ be the event that all these ways fail to satisfy $C$; similarly there are $1+(d-1)(k-1)$ possible ways for $\omega$ to assign variables in $C$, and let $B(\omega)$ be the event that all these ways fail to satisfy $C$. It is easy to see that, there is exactly one common assignment for $\tau,\omega$, where $\tau(x_i)=\omega(x_i)=y$ and $\tau(x_j)=\omega(x_j)=z$.
\begin{align*}
&\mathbf{P}[\tau,\omega\models C]
=\mathbf{P}[C\text{ is satisfied by both }\tau\text{ and }\omega]\\
=&\mathbf{P}(\overline{A(\tau)}\cap\overline{B(\omega)})\\
=&1-\mathbf{P}(A(\tau))-\mathbf{P}(B(\omega))+\mathbf{P}(A(\tau)\cap B(\omega))\\
=&1-\mathbf{P}(A(\tau))-\mathbf{P}(B(\omega))+\mathbf{P}(A(\tau))\mathbf{P}( B(\omega)|A(\tau))\\
=&1-2q^{1+(d-1)(k-1)}+q^{1+2(d-1)(k-1)}.
\end{align*}
\end{proof}

From Lemma \ref{intersection}, we see that
\begin{align*}
&\mathbf{P}(\overline{R_i(\sigma)}\cap\overline{R_j(\sigma)}|S(\sigma)\cap E(e))\\
=&\mathbf{P}\left(\cap_{y\in D_i}\overline{S(\tau)}\cap_{z\in D_j}\overline{S(\omega)}|S(\sigma)\cap E(e)\right)\\
=&\left[\mathbf{P}\left(\overline{S(\tau)}\cap\overline{S(\omega)}|S(\sigma)\cap E(e)\right)\right]^{(d-1)^2},
\end{align*}
where the last equation holds because of independency for any pair of values $(y,z)$.

Let $l_{ij}=|\mathcal{C}(x_i)\cap\mathcal{C}(x_j)|$, then 

\begin{align*}
&\mathbf{P}\left(S(\tau)\cup S(\omega)|S(\sigma)\cap E(e)\right)\\
=&\mathbf{P}\left(\text{$\tau$ satisfies all constraints in $\mathcal{C}(x_i)$, or $\omega$ satisfies all constraints in $\mathcal{C}(x_j)$}\right)\\
=&\mathbf{P}(\tau\models\mathcal{C}(x_i))+\mathbf{P}(\omega\models\mathcal{C}(x_j))-\mathbf{P}(\tau\models\mathcal{C}(x_i),\omega\models\mathcal{C}(x_j))\\
=&\rho^{m_i}+\rho^{m_j}-\rho^{m_i+m_j-2l_{ij}}
\left(1-2q^{1+(d-1)(k-1)}+q^{1+2(d-1)(k-1)}\right)^{l_{ij}}\\
\ge&\rho^{\min\{m_i,m_j\}}.\\ 
\end{align*}
Hence
\begin{align}\label{exright}
&\mathbf{P}\left(\overline{S(\tau)}\cap \overline{S(\omega)}|S(\sigma)\cap E(e)\right)\le1-\rho^{\min\{m_i,m_j\}}.
\end{align}

Combing (\ref{exclusion}), (\ref{exleft}) and (\ref{exright}) we have
\begin{eqnarray}\label{lowerupper}
\nonumber&&1-\sum_{i=1}^n\big(1-\rho^{m_i}\big)^{d-1}\\
&\le&\mathbf{P}\left(\cap_{1\le i\le n}R_i(\sigma)|S(\sigma)\cap E(e)\right)\\
&\le&1-\sum_{i=1}^n\big(1-\rho^{m_i}\big)^{d-1}+
\nonumber\sum_{1\le i<j\le n}\left(1-\rho^{\min\{m_i,m_j\}}\right)^{(d-1)^2}.
\end{eqnarray}

\subsection{Supercritical area: $r<\frac{\alpha}{-\ln p}$}
Let $Y$ be the total number of $(1,1)$-super solutions. By (\ref{Tsigma}) and (\ref{lowerupper}),
\begin{align*}
\mathbf{E}[Y]\ge d^n\binom n k^{-m}p^m\sum_{e\in \mathcal{M}}\left(1-\sum_{i=1}^n\big(1-\rho^{m_i}\big)^{d-1}\right).
\end{align*}

Since $m_iq^{1+(d-1)(k-1)}=o(1)$ as $n$ tends to $\infty$, thus
$$\rho^{m_i}=\left(1-q^{1+(d-1)(k-1)}\right)^{m_i}\ge1-m_iq^{1+(d-1)(k-1)}.$$
Note that $q<1$ and $d=n^{\alpha}$ with $\alpha>0$, thus
\begin{align*}
&1-\sum_{i=1}^n\big(1-\rho^{m_i}\big)^{d-1}\\
\ge &1-q^{(d-1)(1+(d-1)(k-1))}\sum_{i=1}^nm_i^{d-1}\\
\ge&1-q^{d-1+(d-1)^2(k-1)}n(rn\ln n)^{d-1}\\
=&1-\exp\{(d-1)\ln q+(d-1)^2(k-1)\ln q+\ln n+(d-1)\ln(rn\ln n)\}\\
=&1-\exp\{(1+o(1))(d-1)^2(k-1)\ln q\}\\
=&1+o(1).
\end{align*}

As a result,  if $\alpha+r\ln p>0$, then $\mathbf{E}[Y]$ tends to infinity since
\begin{align*}
\ln\mathbf{E}[Y]&\ge n\left\{\ln d+r\ln p\ln n\right\}\\
&=n\ln n(\alpha+r\ln p).
\end{align*}

\subsection{Subcritical area: $r>\frac{\alpha}{-\ln p}$}
Let $Y$ be the total number of $(1,1)$-super solutions. By (\ref{Tsigma}) and (\ref{lowerupper}),
\begin{align*}
&\mathbf{E}[Y]\le d^np^m\mathbf{P}\left(\bigcap_{i=1}^nR_i(\sigma)|S(\sigma)\cap E(e)\right)\\
\le&d^np^m\left[1-\sum_{i=1}^n\big(1-\rho^{m_i}\big)^{d-1}+
\sum_{1\le i<j\le n}\left(1-\rho^{\min\{m_i,m_j\}}\right)^{(d-1)^2}\right]\\
\le &d^np^m\left[1-\sum_{i=1}^n\big(1-\rho^{m_i}\big)^{d-1}\right]+d^np^m\sum_{1\le i<j\le n}\left(1-\rho^{\min\{m_i,m_j\}}\right)^{(d-1)^2}\\
\equiv &E_1+E_2.
\end{align*}

Note that $m_i,m_j\le rn\ln n$, $d=n^\alpha$ and $q<1$, then
\begin{align*}
m_iq^{1+(d-1)(k-1)}\le&\exp\{\ln n+\ln (r\ln n)+n^\alpha(k-1)\ln q+\ln q\}\\
=&\exp\{(1+o(1))n^\alpha(k-1)\ln q\}=o(1).
\end{align*}

Note that $\rho<1$. In addition, if $nx=o(1)$, then $(1-x)^n\ge1-nx$, thus
\begin{align*}
\rho^{\min\{m_i,m_j\}}\ge\rho^{rn\ln n}=(1-q^{1+(d-1)(k-1)})^{rn\ln n}\ge 1-rn\ln nq^{1+(d-1)(k-1)}.
\end{align*}

Thus
\begin{align*}
&\sum_{1\le i<j\le n}\left(1-\rho^{\min\{m_i,m_j\}}\right)^{(d-1)^2}\\
\le&\frac{n^2}2\left(rn\ln nq^{1+(d-1)(k-1)}\right)^{(d-1)^2}\\
=&\frac12\exp\Big\{(1+o(1))(d-1)^3(k-1)\ln q\Big\}=o(1).
\end{align*}
Therefore if $\alpha+kr\ln p<0$, then
\begin{align*}
\mathbf{E}[Y]\le2d^np^m=2\exp\{(\alpha+r\ln p)n\ln n\}=o(1).
\end{align*}

\section{Conclusion}
In this paper, we study the  $(1,1)$-super solutions of model RB, and our results show that there is a ``threshold'' such that as the constraint density crosses this value, the expected number of $(1,1)$-super solutions goes from $0$ to infinity. Interestingly, this threshold is exactly the same with standard satisfiability threshold obtained in \cite{xu}.

However, a rigorous proof of the best possible lower bound on the ``threshold'' of $(1,1)$-super solution should be based on  the second moment method, which is a challenging future work because of the non-negligible correlations among assignments induced by repairing two or more variables at the same time.

\end{document}